\documentclass[conference, a4paper]{IEEEtran}

\ifCLASSINFOpdf
\else
\fi
\usepackage{mathtools}
\mathtoolsset{showonlyrefs}

\usepackage{algorithm}
\usepackage{algpseudocode}

\usepackage[pdftex]{graphicx}
\usepackage{flushend}

\usepackage{expdlist, mdwlist, desclist}

\RequirePackage{bm} 
\RequirePackage{amsmath}
\RequirePackage{amsfonts}
\RequirePackage{amssymb}
\RequirePackage{mathrsfs}
\RequirePackage{dsfont} 
\RequirePackage[single]{accents} 
\RequirePackage{exscale} 

\newcommand{\rv}[1]{\ensuremath{{\boldsymbol{#1}}}}
\renewcommand{\vec}[1]{\ensuremath{{\underline{#1}}}}
\newcommand{\rvec}[1]{\ensuremath{{\boldsymbol{\underline{#1}}}}}
\newcommand{\mat}[1]{{\ensuremath{{\mathbf{#1}}}}}

\def\atan2{\operatorname{atan2}}
\renewcommand{\arg}{\operatorname*{arg}}

\renewcommand{\Im}[1][KeinImagTeil]{%
  \ifthenelse{\equal{#1}{KeinImagTeil}}{
    \ensuremath{\operatorname{Im}}%
  }{
    \ensuremath{\operatorname{Im}\!\left\{#1\right\}}%
  }
}
\renewcommand{\Re}[1][KeinRealTeil]{%
  \ifthenelse{\equal{#1}{KeinRealTeil}}{
    \ensuremath{\operatorname{Re}}%
  }{
    \ensuremath{\operatorname{Re}\!\left\{#1\right\}}%
  }
}

\def\T{^\mathrm{T}} 
\DeclareMathOperator{\trace}{trace}

\def\({\left(}
\def\){\right)}

\newcommand{\Gauss}{{\mathcal{N}}}

\usepackage{theorem}
\theoremheaderfont{\bfseries}
\theoremstyle{plain}
{\theorembodyfont{\itshape} \newtheorem{theorem}     {Theorem}}
{\theorembodyfont{\itshape} }
{\theorembodyfont{\itshape}}
{\theorembodyfont{\sffamily}}
{\theorembodyfont{\normalfont} }
{\theorembodyfont{\rmfamily}}
\newenvironment{proof}
{\noindent{\scshape Proof.}}
{\hspace*{\fill}$\square$}

\def\CON{{CONVEX}}
\def\BBC{{BBC}}
\def\BBL{{BBL}}
\def\BBZ{{BBZ}}
\def\GRE{{GREEDY}}
\def\GRM{{GREEDY*}}
\def\BB{branch-and-bound}

\def\Sec#1{Sec.~\ref{#1}}

\def\Fig#1{Fig.~\ref{#1}}

\def\Theo#1{Theorem~\ref{#1}}
\def\Alg#1{Algorithm~\ref{#1}}

\algrenewcommand{\algorithmiccomment}[1]{\hfill// #1}

\def\x{\rv x}
\def\y{\rv y}

\def\vc{\vec c}

\def\vu{\vec u}

\def\vx{\vec x}

\def\rvv{\rvec v}
\def\rvw{\rvec w}
\def\rvx{\rvec x}

\def\rvz{\rvec z}

\def\A{\mat A}

\def\C{\mat C}

\def\H{\mat H}
\def\I{\mat I}

\def\M{\mat M}
\def\P{\mat P}

\def\hvx{\hat{\vec x}}

\def\hvz{\hat{\vec z}}

\def\SS{\mathcal S}
\def\SU{\mathcal U}

\usepackage[squaren,thinspace]{SIunits}

\newcommand{\Second}[2][]{\unit{#2}#1\second}

\begin{document}
%
\title{Convex Optimization Approach to Multi-Step Sensor Scheduling}
\title{On Multi-Step Sensor Scheduling via \\Convex Optimization}

\author{\IEEEauthorblockN{Marco F. Huber}
\IEEEauthorblockA{%
Variable Image Acquisition and Processing Research Group (VBV)\\
Fraunhofer Institute of Optronics, System Technologies and Image Exploitation IOSB\\
Karlsruhe, Germany\\
Email: marco.huber@ieee.org}
}



\maketitle

\begin{abstract}
Effective sensor scheduling requires the consideration of long-term effects and thus optimization over long time horizons. 
Determining the optimal sensor schedule, however, is equivalent to solving a binary integer program, which is computationally demanding for long time horizons and many sensors. 
For linear Gaussian systems, two efficient multi-step sensor scheduling approaches are proposed in this paper. The first approach determines approximate but close to optimal sensor schedules via convex optimization. The second approach combines convex optimization with a \BB~search for efficiently determining the optimal sensor schedule. 
\end{abstract}


%
\IEEEpeerreviewmaketitle

\section{Introduction}
\label{sec:intro}
%
The recent advances in miniaturization, wireless communication, and sensor technology make it possible to build up and deploy sensor systems 
for a smart and persistent surveillance. For instance, sensor networks consisting of numerous inexpensive sensors are a popular subject in research and practice for monitoring physical phenomena including, 
temperature and humidity distributions, biochemical concentrations, or vibrations in buildings~\cite{Akyildiz_2002}. 
For many of such sensor systems it is necessary to balance between maximizing the information gain and minimizing the consumption of limited resources like energy, computing power, or communication bandwidth.
\emph{Sensor scheduling}, which is also referred to as sensor selection, 
allows trading off these conflicting goals and forms the basis for an efficient and intelligent processing of the sensor data. 

In this paper, the sensor scheduling problem for linear Gaussian systems is studied, where one out of a set of sensors is selected at each time instant for performing a measurement. 
The main objective is to allocate the sensors in a most informative way, which requires making decisions involving multiple time steps ahead. 
Many of the existing multi-step sensor scheduling approaches for linear Gaussian problems are focused on efficiently traversing a decision tree of sensor schedules. 
In order to avoid enumerating all possible sensor schedules of the tree, optimal or suboptimal search techniques are employed. While optimal techniques yield the optimal sensor schedule by all means (see e.g. \cite{Logothetis1999, Fusion08_Huber_PLSS}), 
suboptimal methods as those in \cite{Gupta_ICASSP2004} allow more significant savings in computational demands by abdicating the guarantee of conserving the optimal schedule.

Alternatively to traversing the decision tree, which corresponds to solving a binary integer program, convex optimization approaches have recently been proposed for solving sensor selection problems, i.e., problems of selecting the best $n$-element subset from a set of sensors (see \cite{Chhetri_ISP2007, Joshi_TSP2009}). 
These approaches can significantly improve the efficiency of determining informative sensor schedules, but they are so far not appropriate for
\emph{optimal multi-step} sensor scheduling for \emph{arbitrary} linear Gaussian dynamics and sensor models.  

Both multi-step sensor scheduling approaches proposed in this paper overcome these restrictions. 
For linear Gaussian systems, the sensor scheduling problem is stated in \Sec{sec:problem} and is formulated as a binary integer problem in \Sec{sec:bip}. 
In \Sec{sec:convex} it is formally proven that this optimization problem is a convex optimization problem when employing continuous relaxation of the decision variables. The first approach directly solves the resulting convex program, which leads to suboptimal but valuable sensor schedules without demanding many computations and memory. 
In order to provide the optimal sensor sequence, the second approach described in \Sec{sec:optimal} utilizes \BB~search for traversing a decision tree. To exclude complete subtrees containing suboptimal sensor schedules as early as possible, the solution of the convex optimization is used for calculating tight lower and upper bounds to the subtrees' values. 
The performance of the proposed approaches is demonstrated by means of simulations in \Sec{sec:sim}, while in \Sec{sec:conclusions} conclusions 
and an outlook to future work are given.

\section{Problem Formulation}
\label{sec:problem}
In this paper, the sensor scheduling problem for discrete-time linear Gaussian models is examined.
The dynamics model of the observed system is given by
\begin{equation}
	\label{eq:problem_system}
	\rvx_{k+1} = \A_k\cdot \rvx_k + \rvw_k~,
\end{equation}
where $k = 0, 1, \ldots, $ is the discrete time index. A finite set $\mathcal S$ of sensors is considered for performing measurements, where measurement $\rvz_k^i$ from sensor $i \in \mathcal S = \{1, \ldots, S\}$ is related to the system state $\rvx_k$ via the measurement model
\begin{equation}
	\label{eq:problem_sensor}
	\rvz_k^i = \H_k^i \cdot \rvx_k + \rvv_k^i~.
\end{equation}
Both $\A_k$ and $\H_k^i$ are time-variant matrices. The noise terms $\rvw_k$ and $\rvv_k^i$ are zero-mean white Gaussian with covariance matrices $\C_k^w$ and $\C_k^{v,i}$, respectively. A measurement value $\hvz_k^i$ of sensor $i \in \mathcal S$ is a realization of $\rvz_k^i$. 
The initial system state $\rvx_0 \sim \Gauss(\vx_0; \hvx_0, \C_0^x)$ at time step $k=0$ is Gaussian with mean $\hvx_0$ and covariance~$\C_0^x$.

The aim of multi-step sensor scheduling is to minimize the state covariance $\C_k^x$ and thus the uncertainty of the state estimate under the consideration of the future behavior of the observed dynamical system and long-term sensing costs. For this purpose, 
the optimal sensor schedule $\vu_{1:N}^* = \bigl[ \(\vu_{1}^*\)\T, \ldots, \(\vu_{N}^*\)\T\bigr]\T \in \{0, 1\}^{S\cdot N}$ has to be determined over a finite $N$-step time horizon. Here, the binary vector $\vu_{k}^* = [u_{k,1}, \ldots, u_{k,S}]\T$ encodes the index of the sensor to be scheduled for measurement at time step $k$, i.e., if sensor $i$ is scheduled at time step $k$ then $u_{k,i} = 1$ and $u_{k,j} = 0$ for all $j \neq i$. 
%



\section{Constrained Optimization}
\label{sec:bip}
To determine the optimal sensor schedule $\vu_{1:N}^*$ over a time horizon of length $N$, the sensor scheduling problem for the problem setting given in \Sec{sec:problem} can generally be formulated as constrained optimization problem according to
%
\begin{align}
	\label{eq:problem_optimization}
	\vu_{1:N}^* = \arg \min_{\vu_{1:N}} & \ J(\vu_{1:N})
	\\
	\label{eq:problem_cost}
	\text{subject to} & \ \sum_{k=1}^N \ \vec c_k\T \cdot \vu_k \le C~,
	\\
	\label{eq:problem_one}
	& \ \vec 1\T \cdot \vu_k = 1~, \quad k = 1, \ldots, N~,
	\\
	\label{eq:problem_binary}
	& \ \vu_k \in \{0, 1\}^S~, \quad k = 1, \ldots, N~,
\end{align}
where $\vec c_k = [c_{k,1}, \ldots, c_{k,S}]\T$ contains the sensor costs $c_{k,i}$, e.g., energy or communication, of selecting sensor $i$ at time step~$k$ and $J(\vu_{1:N}) = \sum_{k=1}^N \ g_k\hspace{-.5mm}\( \vu_{1:k} \)$ is the cumulative \emph{objective function} to be minimized. The scalar functions $g_k(\cdot)$, i.e., the summands of $J(\vu_{1:N})$, can be
\begin{itemize}
	\item the trace operator $\trace \(\C_k^x(\vu_{1:k})\)$~,
	\item the root-determinant $\sqrt{|\C_k^x(\vu_{1:k})|}$~,
	\item or the maximum eigenvalue $\lambda_{\mathrm{max}}\(\C_k^x(\vu_{1:k})\)$
\end{itemize}
of the state covariance and thus quantify the uncertainty subsumed in $\C_k^x(\vu_{1:k})$. 
Due to the linearity of \eqref{eq:problem_system} and \eqref{eq:problem_sensor}, the state covariance itself is given by the Kalman filter covariance recursion (given here in information form)
\begin{multline}
	\label{eq:problem_cov}
	\C_k^x(\vu_{1:k}) = \Bigl( \bigl(\A_{k-1} \cdot \C_{k-1}^x(\vu_{1:k-1}) \cdot \A_{k-1}\T+\C_{k-1}^w\bigr)^{-1}
	\\
	+\sum_{i=1}^S\ u_{k,i} \cdot \bigl(\H_k^i\bigr)\T \cdot \bigl(\C_k^{v,i}\bigr)^{-1} \cdot \H_k^i \Bigr)^{-1}~,
\end{multline}
commencing from $\C_0^x$. The independence of the covariance recursion from the measurement values $\hvz_k^i$, $i=1,\ldots,\cal S$, allows the off-line calculation of the optimal sensor schedule.

With the contraint in \eqref{eq:problem_cost} 
it is guaranteed that a feasible sensor schedule does not exceed a maximum sensor cost $C$. The constraints in \eqref{eq:problem_one} and \eqref{eq:problem_binary} together ensure that one sensor per time step is selected for measurement. 
The restriction to one sensor per time step is made for brevity and clarity reasons only. 
The extension to selecting multiple sensors per time step can be achieved by replacing the right hand side of \eqref{eq:problem_one} with the desired number of sensors.
Alternatively, by modifying \eqref{eq:problem_optimization} and \eqref{eq:problem_cost},
is also possible to minimize the sensor costs regarding a maximum allowed value of $J(\cdot)$, i.e., a maximum allowed uncertainty. The results presented in the following can be easily altered to cover this modified optimization problem.

\section{Convex Relaxation}
\label{sec:convex}
%
%
The optimization problem in \eqref{eq:problem_optimization}--\eqref{eq:problem_binary} is a so-called \emph{binary integer program}. Problems of this type are known to be NP-hard (see \cite{Karp1972}) and thus, obtaining the optimal solution for large $N$ and/or large $S$ is computationally prohibitive in general. However, by replacing the binary non-convex constraints in \eqref{eq:problem_binary} with the linear constraints $\vu_{k} \in [0,1]^S$ for $k = 1, \ldots, N$, a convex relaxation of the original problem is obtained. To see this, it is important to note that the constraints \eqref{eq:problem_cost} and \eqref{eq:problem_one} are already convex. Furthermore, as shown in the following theorem, the sum to be minimized in \eqref{eq:problem_optimization} is now convex as well.

\begin{theorem}
\label{theo:convex}
The objective function $J(\vu_{1:N})$ in \eqref{eq:problem_optimization} is convex in terms of $\vu_{1:N}\in [0,1]^{S\cdot N}$. 
\end{theorem}
\begin{proof}
To prove the convexity of $g_k(\vu_{1:k})$ and thus of $J(\vu_{1:N})$, it must be shown that (see for example \cite{Boyd2008})
\begin{multline}
	\label{eq:convex_proof}
	g_k\hspace{-.5mm}\(\lambda\cdot \vu_{1:k} + (1\hspace{-.5mm}-\hspace{-.5mm}\lambda)\cdot \tilde \vu_{1:k}\) 
  \ \le \\
	\lambda\cdot g_k\hspace{-.5mm}\(\vu_{1:k}\) + (1\hspace{-.5mm}-\hspace{-.5mm}\lambda)\cdot g_k\hspace{-.5mm}\(\tilde \vu_{1:k}\)
\end{multline}
for $k = 1,\ldots, N$, $\forall\, \vu_{1:k}, \tilde \vu_{1:k} \in [0, 1]^{k\cdot S}$, and $\forall\, \lambda\in [0,1]$.

At first, it is proven by induction that the covariance recursion \eqref{eq:problem_cov} is a convex function of $\vu_{1:k}$.  
The induction starts with $\C_1^x(\vu_1)$. Defining $\M_k^i := \(\H_k^i\)\T \cdot \bigl(\C_k^{v,i}\bigr)^{-1} \cdot \H_k^i$ and $\P_1(\vu_1) := \(\A_0 \cdot \C_0^x \cdot \A_0\T + \C_0^w\)^{-1} + \sum_i u_{1,i} \cdot \M_1^i$ and utilizing the results in \cite{Kraus1936} on matrix convex functions, it follows from the matrix convexity property of the matrix inversion that
\begin{align}
	\C_1^x(\lambda\cdot \vu_1+(1\hspace{-.5mm}-\hspace{-.5mm}\lambda)\cdot \tilde \vu_1) 
	&= \(\lambda\cdot \P_1(\vu_1) + (1\hspace{-.5mm}-\hspace{-.5mm}\lambda)\cdot \P_1(\tilde \vu_1)\)^{-1} \\[-7mm]
	& \le \lambda\cdot \underbrace{\P_1^{-1}(\vu_1)}_{= \C_1^x(\vu_1)} 
	+ (1\hspace{-.5mm}-\hspace{-.5mm}\lambda)\cdot \underbrace{\P_1^{-1}(\tilde \vu_1)}_{= \C_1^x(\tilde \vu_1)}\\[-9mm]
\end{align}
$\forall\, \vu_1, \tilde \vu_1 \in [0, 1]^{S}$ and $\forall\, \lambda\in [0,1]$. 
Defining the predicted covariance $\C_k^p(\vu_{1:k-1}) := \A_{k-1} \cdot \C_{k-1}^x(\vu_{1:k-1})\cdot\A_{k-1}\T + \C_{k-1}^w$, it generally holds that
%
\begin{align}
	\C_k^x&(\lambda\cdot \vu_{1:k} + (1\hspace{-.5mm}-\hspace{-.5mm}\lambda)\cdot \tilde \vu_{1:k})\\
	&=\Bigl( \C_{k}^p\hspace{-.7mm}\(\lambda\cdot\vu_{1:k-1} + (1\hspace{-.5mm}-\hspace{-.5mm}\lambda)\cdot\tilde \vu_{1:k-1}\)^{-1}\\ 
	&\hspace{1cm}+ \sum_{i=1}^S\, (\lambda\cdot u_{k,i} + (1\hspace{-.5mm}-\hspace{-.5mm}\lambda)\cdot\tilde u_{k,i}) \cdot \M_k^i \Bigr)^{\hspace{-.2mm}-1}\\
	&\stackrel{(a)}{\le} \biggl( \lambda\cdot\Bigl(\C_{k}^{p}\hspace{-.7mm}\(\vu_{1:k-1}\)^{-1} + \sum_{i=1}^S\ u_{k,i}\cdot \M_k^i\Bigr) \\
	&\hspace{1cm}+ (1\hspace{-.5mm}-\hspace{-.5mm}\lambda)\cdot\Bigl( \C_{k}^{p}\hspace{-.7mm}\(\tilde\vu_{1:k-1}\)^{-1} + \sum_{i=1}^S\ \tilde u_{k,i}\cdot \M_k^i\Bigr) \biggr)^{-1}\\[-4mm]
	\label{eq:convex_Cn}
	&\stackrel{(b)}{\le} \lambda\cdot \C_k^x(\vu_{1:k}) + (1\hspace{-.5mm}-\hspace{-.5mm}\lambda)\cdot \C_k^x(\tilde \vu_{1:k})
\end{align}
for $k = 2,\ldots, N$, $\forall\, \vu_{1:k}, \tilde \vu_{1:k} \in [0, 1]^{k\cdot S}$, and $\forall\, \lambda\in [0,1]$. Here, (a) results from the induction hypothesis that $\C_{k-1}^x(\vu_{1:k-1})$ is convex in $\vu_{1:k-1}$, from the convexity of the matrix inversion, and from rearranging terms; (b) is the result of a repeated application of the convexity of the matrix inversion.

As the trace is a linear matrix function and the root-determinant as well as the maximum eigenvalue are convex matrix functions (see for example \cite{Boyd2008}), the inequality in \eqref{eq:convex_proof} holds if these three functions are applied on \eqref{eq:convex_Cn}. Thus, $g_k\hspace{-.5mm}\(\vu_{1:k}\)$ is convex and the nonnegative sum  $J(\vu_{1:N}) = \sum_{k=1}^N g_k(\vu_{1:k})$ is convex as well, which concludes the proof.
\end{proof}

\Theo{theo:convex} forms one of the main contributions of this paper. It is important to note that the sensor scheduling problem formulated by \eqref{eq:problem_optimization}--\eqref{eq:problem_cov} and its convex relaxation proven in \Theo{theo:convex} extends existing convex approaches \cite{Chhetri_ISP2007, Joshi_TSP2009} in many ways. Instead of one-step time horizons, i.e., myoptic/greedy scheduling, arbitrarily long time horizons are possible. Furthermore, the dynamics model in \eqref{eq:problem_system} need not to be restricted to regular system matrices $\A_k$ and to system noise covariances $\C_k^w = \mat 0$. Especially the latter is of paramount importance for realistic sensor scheduling problems. Finally, 
there is no restriction to a specific scalar function $g_k(\cdot)$ as in \cite{Chhetri_ISP2007}. Instead, various functions for evaluating the quality of a sensor schedule are considered here.

\subsection{Solving the Relaxed Problem}
\label{sec:convex_solving}
The computational complexity of optimally solving the original binary integer program is in $\mathcal O(S^N)$. 
A variety of methods is available for efficiently solving the convex relaxation of the sensor scheduling problem, e.g., interior-point methods \cite{Boyd2008}. These methods typically require only a few tens of iterations for calculating the optimal solution even for large problem sizes, e.g., length of time horizon and number of sensors beyond $10$. The computational complexity of one iteration is polynomial in the number of variables in $\vu_{1:N}$, which is~$S\cdot N$.

However, the solution $\vu_{1:N}^\text{l}$ of the convex problem only approximates the optimal solution $\vu_{1:n}^*$ of the original scheduling problem. More specifically, $\vu_{1:N}^\text{l}$ is no longer binary and the objective function value $J^\text{l} := J(\vu_{1:N}^\text{l})$ is a \emph{lower bound} of the optimal value $J(\vu^*_{1:N})$. The latter finding follows directly from the convexity of the relaxed problem and from the fact that the relaxed solution set $[0,1]^{S\cdot N}$ contains the binary set of the original problem.

\subsection{Conversion into Binary Solution}
\label{sec:convex_conversion}
In order to allow selecting sensors for measurement, $\vu_{1:N}^\text{l}$ has to be converted into a binary vector by employing an appropriate conversion or rounding method.
The value $J^\text{u} := J(\vu_{1:N}^\text{u})$ of the resulting (binary) sensor schedule $\vu_{1:N}^\text{u}$ has to be as close as possible to the optimal one in order to provide informative sensor measurements. 
In the following, two appropriate conversion methods are introduced. Independent of the chosen conversion method, the value $J^\text{u}$ of the converted sensor schedule provides an \emph{upper bound} to the optimal value~$J(\vu^*_{1:N})$.

\subsubsection{Sampling}
\label{sec:convex_conversion_sampling}
Each component $\vu_k^\text{l}$ of $\vu_{1:N}^\text{l}$ can be interpreted as a discrete probability distribution over the set of sensor indices $\mathcal S$. This is due to 
the constraint in \eqref{eq:problem_one}, whereby the elements $u_{k,i}^\text{l}, i = 1,\ldots, S$ of $\vu_k^\text{l}$ are within the interval $[0,1]$ and sum up to one.
Hence, a sensor $i$ corresponding to an element $u_{k,i}^\text{l}$ with a large value can be considered as being more likely in the optimal sensor schedule than sensors with small values. 

To convert $\vu_{1:N}^\text{l}$ into a feasible binary vector, for each $k=1,\ldots,N$ a (single) sensor is randomly selected according to the distribution $\vu_k^\text{l}$. For being \emph{feasible}, the resulting converted schedule $\vu_{1:k}^\text{u}$ has to satisfy the cost constraint \eqref{eq:problem_cost}. Otherwise, the schedule is discarded. This procedure is repeated multiple times, where only the currently best feasible schedule, i.e., the schedule that satisfies \eqref{eq:problem_cost} \emph{and} provides the currently smallest objective function value $J^\text{u}$ is stored. 
The sampling-based conversion method can be terminated for example after a predefined number of trials or when the currently best value $J^\text{u}$ remains unchanged for a predefined number of trials.

\subsubsection{Swapping}
\label{sec:convex_conversion_swapping}
A converted schedule $\vu_{1:N}^\text{u}$ can be improved by adapting the swapping method proposed in \cite{Joshi_TSP2009}. A modified sensor schedule is derived from $\vu_{1:N}^\text{u}$ by swapping a scheduled sensor with one of the unselected sensors for each time step $k$. 
The choice of an unselected sensor at time step $k$ is \emph{deterministically} guided according to the probabilities represented by $\vu_k^\text{l}$, i.e., the sensors are selected for swapping in descending order of the values in $\vu_k^\text{l}$. If the modified schedule is feasible and improves the objective function value $J^\text{u}$, it is used as starting point for the next swapping trial. 

In order to start the swapping method with a feasible schedule, the sensor schedule that selects at each time step $k$ the sensor $i=\arg \min_{j} c_{k,j}$ with the smallest cost is chosen initially. 
The method must terminate because there is only a finite but very large number of swapping possibilities. To bound the computational demand, the number of swapping trials is limited by means of a predefined value.

\section{Optimal Scheduling}
\label{sec:optimal}
Determining the optimal sensor schedule and thus, directly solving the binary integer program given by \eqref{eq:problem_optimization}--\eqref{eq:problem_binary} can be considered as searching a decision tree with depth $N$ and branching factor $S$. The problem here is that the optimal solution often can be found at an early stage, while the proof of its optimality requires evaluating most of the suboptimal sensor schedules, which is infeasible for large problem sizes.  
In this section, the previously introduced convex optimization approach is combined with efficient search methods for decision trees for early eliminating (pruning) suboptimal schedules. 

\subsection{Branch-and-Bound}
\label{sec:optimal_bab}
A search technique common for classical decision problems like traveling-salesman or knapsack is \BB~(BB) search. The basic idea of BB is to assign lower and upper bounds of the achievable objective function value to any visited node. Based on these bounds, nodes and thus complete subtrees can be pruned under the guarantee that the pruned node is not part of the optimal sensor schedule.

For a particular node that was reached during the search by employing the sensor schedule $\vu_{1:k} \in \{0, 1\}^{k\cdot S}$, the objective function can be written according to
\begin{align}
	\label{eq:optimal_objective}
	J(\vu_{1:N}) = \underbrace{J(\vu_{1:k})}_\text{known} + \underbrace{J(\vu_{k+1:N})}_\text{unknown}~,
\end{align}
where only the value of the first summand is evaluated and thus known. While the value of the second summand is not calculated yet, a lower and upper bound can be easily assigned to it by exploiting the results of \Sec{sec:convex_solving} and \Sec{sec:convex_conversion}. The value of the optimal solution $\vu_{k+1:N}^\text{l}$ of the convex relaxation for minimizing $J(\vu_{k+1:N})$ serves as lower bound and the conversion of $\vu_{k+1:N}^\text{l}$ into a binary-valued vector $\vu_{k+1:N}^\text{u}$ provides an upper bound. Hence, the inequality
\begin{align}
	J(\vu_{1:k}) +  J(\vu_{k+1:N}^\text{l}) \le J(\vu_{1:N}) \le J(\vu_{1:k}) + J(\vu_{k+1:N}^\text{u})
\end{align}
holds for the objective function value in \eqref{eq:optimal_objective}.

\begin{algorithm}[tb]
\caption{Branch-and-Bound search algorithm utilizing convex optimization for calculating lower and upper bounds. The algorithm is initialized with $J_\mathrm{min} = \infty$.}
\label{alg:bb}
\begin{algorithmic}[1]
\State For a given sensor schedule $\vu_{1:k}$ do:
\If{leaf node, i.e., $k=N$}
	\State $J_\mathrm{min} \leftarrow J(\vu_{1:N})$ 
\Else
	\State $\SU \leftarrow \emptyset$  \algorithmiccomment{List of sensors to expand}
	\ForAll{sensors $i\in \{1, \ldots, S\}$} 
		\State // $\vu_{1:k}$ and $u_{k+1,i}=1$ fixed
		\If{cost$_i \le C$ and $J_i \le J_\mathrm{min}$}
		\State $J_i^\text{l} \leftarrow$ Solve convex optimization problem
		\State $J_i^\text{u} \leftarrow$ Calculate upper bound via conversion
		\State $\SU \leftarrow \SU \cup \{i\}$
		\EndIf
	\EndFor
		\State $\SU \leftarrow$ sort($\SU$) \algorithmiccomment{Sort sensors based on $J_i^\text{l}$}
		\ForAll{sensors $i \in \SU$}
			\If{$J_i^\text{l} \le  J_\mathrm{min}$ and $\forall\, j \in \SU: J_i^\text{l} \le J_j^\text{u}$}
				\State Expand $i$\algorithmiccomment{Set $u_{k+1,i} = 1$, call Algorithm 1}
			\EndIf
		\EndFor
\EndIf
\end{algorithmic}
\end{algorithm}

\subsection{Search Algorithm}
\label{sec:optimal_pruning}
The combination of BB search with convex optimization is illustrated in \Alg{alg:bb}, which basically employs a depth-first search. 
For a given sensor schedule $\vu_{1:k}$ it is checked which child nodes should be expanded, i.e., it is checked whether an element $u_{k+1, i}$, $i\in \SS$ of $\vu_{k+1}$ could be set to one or not. Therefore, for each child node $i\in \SS$ the minimum possible cost \vspace{-2mm}
\begin{equation}
	\text{cost}_i := \sum_{n=1}^k \vc_n \cdot \vu_{n} + c_{k+1,i} + \sum_{n=k+2}^N \min_j c_{n,j}~,\vspace{-2mm}
\end{equation}
the value $J_{i} := J(\vu_{1:k+1})$ and the bounds $J_i^\text{l} := J_i + J(\vu_{k+2:N}^\text{l})$, $J_i^\text{u} := J_i + J(\vu_{k+2:N}^\text{u})$ are calculated, where $u_{k+1,i} = 1$ and $u_{k+1,j} = 0$ for all $j \neq i$. 
A node $i$ is expanded only if following four requirements are fulfilled: (1) the cost constraint can be met, i.e., a feasible solution exists (line~8), (2) the value $J_{i}$ of the node is below the value $J_\mathrm{min}$ of the currently best sensor schedule (line~8), (3) the lower bound $J_i^\text{l}$ is below $J_\mathrm{min}$ (line~16), and (4) the lower bound is below the upper bounds of all neighboring nodes $j\neq i$ (line~16).


Obviously, the third requirement implies the second one. But in order to avoid an unnecessary calculation of the lower and upper bound, the second requirement is checked separately together with the first requirement (line~8--12). 
To further accelerate the search, the remaining sensors in $\SU$ are sorted in descending order according of their lower bounds (line~12). In doing so, the search is continued with the most promising sensor first in order to force a stronger reduction of the currently best value $J_\mathrm{min}$. This value is automatically reduced once a leaf node is reached (line~2--3).

\section{Simulation Results}
\label{sec:sim}

\begin{figure*}[t]%
\centering
\includegraphics[]{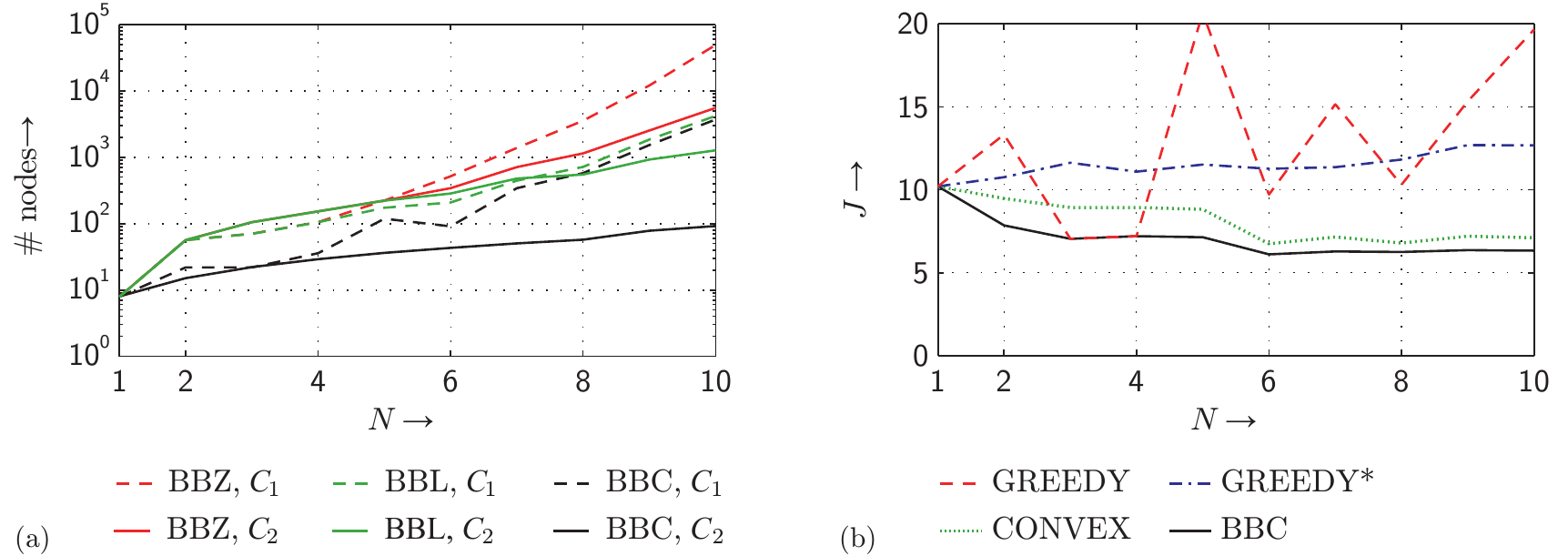}%
\caption{(a) Number of nodes in the decision tree when applying the \BB-based scheduling methods \BBC~(black lines), \BBL~(green), and \BBZ~(red) for different time horizons lengths $N$ and for the cost functions $C_1$ (dashed) and $C_2$ (solid) in log-scale. (b) Objective function values $J$ of the scheduling methods \BBC~(black, solid), \CON~(green, dotted), \GRE~(red, dashed), and \GRM~(blue, dash-dotted) for maximum cost function~$C_1$.}
\vspace{-1mm}
\label{fig:sim}%
\end{figure*}

The effectiveness of the proposed sensor scheduling \linebreak methods is demonstrated in the following by means of a numerical simulation from the field of target tracking\footnote{Further aspects associated to target tracking such as target detection, misses and false alarms, or track maintenance are omitted for simplicity.}. The state $\rvx_k = [\x_k, \dot{\x}_k, \y_k, \dot{\y}_k]\T$ of the observed target comprises the two-dimensional position $[\x_k, \y_k]\T$ and the velocities $[\dot{\x}_k, \dot{\y}_k]\T$ in $x$ and $y$ direction. The 
system matrix and noise covariance matrix of $\rvw_k$ of the dynamics model \eqref{eq:problem_system} are 
\begin{align}
	\label{sec:sim_dynamics}
	\A_k = \I_2 \otimes \begin{bmatrix} 1 & T \\ 0 & 1\end{bmatrix}
	\ \text{ and }\ 
	\C_k^w = q \cdot \I_2 \otimes 
	\begin{bmatrix} 
		\tfrac{T^3}{3} &\tfrac{T^2}{2} \\ \tfrac{T^2}{2} & T
	\end{bmatrix}
	~,~\quad
\end{align}
respectively, where $\I_n$ indicates an $n \times n$ identity matrix and $\otimes$ is the Kronecker matrix product. In \eqref{sec:sim_dynamics},  $T=\Second{1}$ is the sampling interval and $q=0.2$ is the scalar diffusion strength. Mean and covariance of the initial state $\rvec x_0$ are $\hat{\vx}_0 = [0, 1, 0, 1]\T$ and $\C_0^x = 10\cdot \I_4$, respectively.

A sensor network observes the target. It consists of six sensors with measurement matrices 
\begin{align}
	\H_k^1 = \H_k^3 &=
	\begin{bmatrix}
		1 & 0 & 0 & 0
	\end{bmatrix}~,~
	\H_k^2 = \H_k^5 =
	\begin{bmatrix}
		0 & 0 & 1 & 0
	\end{bmatrix}~,
	\\
	\H_k^4 &=
	\begin{bmatrix}
		0 & 0 & 0 & 1
	\end{bmatrix}~,~
	\H_k^6 =
	\begin{bmatrix}
		0 & 1 & 0 & 0
	\end{bmatrix}~,
\end{align}
noise variances $C_k^{v,1} = 0.2$, $C_k^{v,2} = C_k^{v,3} = C_k^{v,4} = 0.1$, $C_k^{v,5} = C_k^{v,6} = 0.05$, and costs $\vec c_k = [1, 2, 3, 2, 3, 2]\T$ for each~$k$. Furthermore, it is also possible to omit a measurement. This option can be considered as having a seventh sensor with infinite noise variance. Performing no measurement is free of cost, i.e., $c_{k,7} = 0$.  Altogether, the set $\mathcal S$ comprises $S=7$ sensors. The scalar functions $g_k(\cdot)$ are set to the root-determinant for each $k$.

For comparison, six different scheduling methods are considered: 
\begin{description}
	\item[\CON] The approach described in \Sec{sec:convex}, which directly solves the convex optimization problem and employs the swapping method for conversion. 
	\item[\BBC] The BB~approach described in \Sec{sec:optimal}. For determining the upper bounds via conversion, the swapping method is employed.
	\item[\BBL] Like \BBC~but without utilizing upper bounds.
	\item[\BBZ] BB search that employs no upper bounds and bounds the second summand in \eqref{eq:optimal_objective} from below with zero (see for example \cite{Chhetri_EURASIP2006}). 
	\item[\GRE] In order to minimize the objective function $J(\cdot)$, at each time step $k$ the sensor that leads to minimum function value $g_k(\cdot)$ \emph{and} allows meeting the maximum cost constraint is scheduled.
	\item[\GRM] Greedy scheduling, where the scalar functions $g_k(\cdot)$ are modified to $g_k(\vu_{1:k}) = \sqrt{|\C_k^x(\vu_{1:k})|}\cdot(1 + \vec c_k\T\cdot \vu_k)$ (see for example \cite{Zhao2002, Khuller_IPL1999}).
\end{description}
For \CON~and \BBC, the number of swapping trials is set to $S\cdot N$.

\subsection{Comparison of Branch-and-Bound Methods}
In \Fig{fig:sim}~(a), the search performance of the three BB methods is compared. For this purpose, two different maximum costs $C_1(N) = \mathrm{round}(1.5\cdot N)$ and $C_2(N) = 3\cdot N$ are considered, which depend on the change of the time horizon length $N = 1, \ldots, 10$. The maximum cost function $C_2(N)$ allows sensor scheduling without omitting a measurement. With the proposed optimal scheduling method \BBC, the number of nodes in the decision tree can be kept on a low level.
Here, the search performance clearly benefits from the tight lower and upper bounds provided by the convex optimization and the conversion, respectively. This can be seen in particular for $C_2$, where \BBC~only visits at most $92$ nodes, while the complete decision tree contains $\sum_{k=1}^N 7^k<3.3\cdot 10^8$ nodes.
The higher number of visited nodes for cost function $C_1$ compared to $C_2$ results from the effect that the more restrictive cost constraint provided by $C_1$ leads to looser bounds. 

Without considering upper bounds for pruning as it is the case for \BBL, the number of visited nodes increases significantly. But still, the search performance of \BBL~is much better than \BBZ~as the lower bound provided by the solution of the convex optimization is closer to the true values of the subtrees. 

Since calculating lower and upper bounds by means of convex relaxation is computationally more demanding than calculating the simple bound used for \BBZ, the runtime of \BBZ~is lower for short time horizons even if \BBZ~leads to larger decisions trees. But with an increasing length of the time horizon, the difference in runtime between \BBZ~and the other BB methods becomes smaller and at some point, both methods outperform \BBZ. For example, with the current, barely optimized implementation based on MATLAB version $7.9$, \BBC~outperforms \BBZ~from horizon length $N=9$ on for cost function $C_1$. It is expected that employing an optimized implementation, e.g., with an optimized convex programming toolbox like CVX \cite{Grant_CVX}, outperforming \BBZ~occurs for significantly shorter time horizons.


\subsection{Comparison of Objective Function Values}
In \Fig{fig:sim}~(b), the objective function values of \BBC~are compared with \GRE, \GRM, and \CON~for the costs $C_1(N)$. Both greedy methods are the computationally cheapest, but in turn provide highly suboptimal results. Due to the myopic planning, \GRE~is not able to anticipate the long-term effect of early selecting costly sensors. In this simulation example, \GRE~omits measurements at the last time steps of the horizon and not in between in order to meet the maximum cost constraint. This effect is attenuated by \GRM~thanks to the modified scalar functions $g_k(\cdot)$, where the sensor costs are minimized together with the state covariance. 
In doing so, \GRM~can for example utilize sensor $1$ instead of the more costly but also more accurate sensor $3$. This in turn enables \GRM~to omit less measurements at the end of the time horizon.

The proposed suboptimal \CON~method provides sensor schedules close to the optimal ones, whereas the computational demand is significantly smaller compared to \BBC, especially for very long time horizons. 
\CON~trades scheduling quality off against scheduling complexity, which is desirable for computationally constrained sensor systems.

\subsection{Tracking Error}
\label{sec:sim_trackingerror}
For a time horizon of length $N=10$ and for the cost function $C_1(N)$, $100$ Monte Carlo simulation runs are performed for evaluating the target tracking error when employing \GRE, \GRM, \CON, and \BBC. In \Fig{fig:sim_rms}, the root mean square error (rmse) with respect to the target position $[x, y]\T$ is depicted. Compared to the greedy methods, both \CON~and \BBC~provide the lowest tracking errors together with the lowest uncertainty (consider \Fig{fig:sim}~(b)), whereas \BBC~performs best. The little bump in the curves of \CON~and \BBC~around time step $k=6$ results from omitting a measurement. Here, \GRM~has the smallest tracking error, which comes at the expense of a higher error at the end of the time horizon due to omitting measurements. 

However, in the simulation example considered here, the tracking error of \GRM~is relatively close to the errors of the proposed convex optimization based approaches. 
More significant benefits of multi-step sensor scheduling are expected for example in scenarios where sensors are temporarily unavailable or in scenarios with nonlinear dynamics and sensor models (consider for example the results in \cite{Chhetri_EURASIP2006, PhD2009}). The application of the proposed approaches to nonlinear models is subject of future work. 


\begin{figure}[t]%
\centering
\includegraphics[]{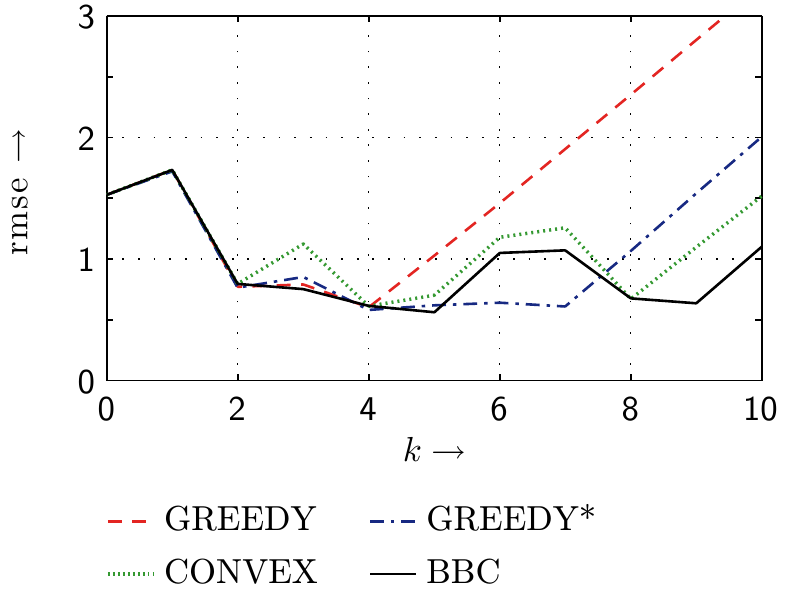}%
\caption{Rmse of the target position $[x, y]\T$ when \GRE~(red, dashed), \GRM~(blue, dash-dotted), \CON~(green, dotted), and \BBC~(black, solid) are employed over a time horizon of length $N=10$.}
\vspace{-3mm}
\label{fig:sim_rms}%
\end{figure}

\section{Conclusions and Future Work}
\label{sec:conclusions}
Convex optimization is a promising direction for determining multi-step sensor schedules. In this paper, a general sensor scheduling problem for linear Gaussian systems was formulated and the convexity of its relaxation was proven. Based on this result, a suboptimal and an optimal scheduling approach utilizing convex optimization have been proposed. 
While the suboptimal approach trades estimation quality off against computational demand, the optimal one outperforms existing optimal algorithms with regard to search speed.


Compared to existing approaches on sensor scheduling via convex optimization, arbitrary linear Gaussian sensor scheduling problems are covered. Furthermore, both proposed scheduling methods are appropriate for long time horizons and many sensors, where choosing the better suited approach for a given scheduling problem depends on the requirements on estimation quality and computational capabilities.

Future work is devoted to efficient sensor scheduling for very long or even infinite time horizons, where the BB-based approach is computationally intractable. Here, model-predictive control (also referred to as moving horizon control, see for example \cite{Camacho2007}) can be employed. Furthermore, it is intended to extend the proposed convex sensor scheduling approaches to nonlinear dynamics and sensor models. 
This could be achieved for instance by employing model-predictive control in combination with a conversion of the nonlinear models into linear ones via linearization, e.g., first-order Taylor series expansion or statistical linearization as proposed in \cite{PhD2009}. 



\bibliographystyle{IEEEtran}

\end{document}